\documentclass{statsoc}
\usepackage{amsmath,epsfig}
\usepackage{inputenc}
\usepackage[T1]{fontenc}
\usepackage[french,english]{babel}
\usepackage{amssymb,color,bbm}
\usepackage{stmaryrd}
\usepackage{dsfont}
\usepackage{graphicx}
\usepackage{subfig}
\usepackage{version}
\usepackage{ifthen}
\usepackage{enumerate}
\usepackage{algorithm,algpseudocode}

\def\Xset{\mathbb{X}}
\def\Yset{\mathbb{Y}}
\def\Xsigma{\mathcal{X}}
\def\Ysigma{\mathcal{Y}}
\def\Aset{\mathsf{A}}

\newcommand{\norm}[2]{\left\| #2\right\|_{#1}}
\newcommand{\Var}[2]{\mathbb{V}\mathrm{ar}_{#1}\left(#2\right)}
\newcommand{\TV}{\mathrm{TV}}

\newcommand{\e}{\mathrm{e}}
\newcommand{\tfcal}{\tilde{\mathcal{F}}}
\newcommand{\fcal}{\mathcal{F}}
\newcommand{\calC}{\mathcal{C}}

\newcommand{\norminf}[1]{|#1|_{\infty}}

\newcommand{\nset}{\mathbb N}
\newcommand{\one}{\bold{1}}
\newcommand{\rmd}{\mathrm{d}}
\newcommand{\rme}{\mathrm{e}}

\newcommand{\eqdef}{\ensuremath{\stackrel{\mathrm{def}}{=}}}
\newcommand{\eqsp}{\;}
\newcommand{\bigo}[1]{\mathcal{O}\left(#1\right)}

\newcommand{\iid}{i.i.d.}

\newcommand{\XinitIS}[2][]{\ifthenelse{\equal{#1}{}}{\ensuremath{\rho_{#2}}}{\ensuremath{\check{\rho}_{#2}}}}

\newcommand{\kiss}[3][]
{\ifthenelse{\equal{#1}{}}{p_{#2}}
{\ifthenelse{\equal{#1}{fully}}{p^{\star}_{#2}}
{\ifthenelse{\equal{#1}{smooth}}{\tilde{r}_{#2}}{\mathrm{erreur}}}}}

\newcommand{\genPart}[3][]{\boldsymbol{\xi}^{#2,#3}\ifthenelse{\equal{#1}{}}{}{[#1]}}
\newcommand{\genRePart}[3][]{\tilde{\boldsymbol{\xi}}^{#2,#3}\ifthenelse{\equal{#1}{}}{}{[#1]}}
\newcommand{\genWeight}[2]{\omega^{#1,#2}}
\newcommand{\genTarget}{\Pi}
\newcommand{\moveKernel}{Q}
\newcommand{\var}[2][]{\mathbb{V}\mathrm{ar}\left( #2 \ifthenelse{\equal{#1}{}}{}{\middle|#1}\right)}
\newcommand{\dlim}{\ensuremath{\stackrel{\mathcal{D}}{\longrightarrow}}}
\newcommand{\plim}{\ensuremath{\stackrel{\mathbb{P}}{\longrightarrow}}}

\newcommand{\osc}[1]{\mathrm{osc}\left(#1\right)}

\newcommand{\post}[3][]%
{
\ifthenelse{\equal{#1}{}}{\ensuremath{\pi_{#2|#3}}}%
{\ifthenelse{\equal{#1}{algo}}{\ensuremath{\pi^{\mathrm{algo},N}_{#2|#3}}}
{\ifthenelse{\equal{#1}{tilde}}{\ensuremath{\pi^{\mathrm{FFBSi},N}_{#2|#3}}}
{\ifthenelse{\equal{#1}{tar}}{\ensuremath{\pi^{N,\mathrm{t}}_{#2|#3}}}}
}
}
}

\newcommand{\Post}[3][]%
{
\ifthenelse{\equal{#1}{}}{\ensuremath{\Pi_{#2|#3}}}%
{\ifthenelse{\equal{#1}{algo}}{\ensuremath{\Pi^{\mathrm{algo},N}_{#2|#3}}}
{\ifthenelse{\equal{#1}{tilde}}{\ensuremath{\Pi^{\mathrm{FFBSi},N}_{#2|#3}}}
{\ifthenelse{\equal{#1}{tar}}{\ensuremath{\Pi^{N,\mathrm{t}}_{#2|#3}}}}
}
}
}

\newcommand{\PE}[2][]{\mathbb{E}\left[ #2 \ifthenelse{\equal{#1}{}}{}{\middle|#1}\right]}

\newcommand{\PP}[1]{\mathbb{P}\left[#1\right]}
\newtheorem{prop}{Proposition}
\newtheorem{theorem}{Theorem}

\newcounter{hypA}
\newenvironment{hypA}{\smallskip \refstepcounter{hypA}\begin{itemize}
  \item[({\bf A\arabic{hypA}})]}{\end{itemize} \smallskip }

\newcommand{\A}[1]{(\textbf{A#1})}

\renewenvironment{proof}%
 {\noindent{\textsc{Proof}.} }%
 { \hfill$\square$}               

\title{Particle approximation improvement of the joint smoothing distribution with on-the-fly  variance estimation}
\author[C. Dubarry {\it et al.}]{Cyrille Dubarry}
\address{TELECOM SudParis\\
        D\'epartement CITI\\
        9 rue Charles Fourrier\\
        Evry, France.\\
        cyrille.dubarry@telecom-sudparis.eu}
\email{cyrille.dubarry@telecom-sudparis.eu}
\author[C. Dubarry {\it et al.}]{Randal Douc}
\address{TELECOM SudParis\\
        D\'epartement CITI\\
        9 rue Charles Fourrier\\
        Evry, France.\\
        randal.douc@telecom-sudparis.eu}
\email{randal.douc@telecom-sudparis.eu}

\begin{document}

\begin{abstract}
Particle smoothers are widely used algorithms allowing to approximate the
smoothing distribution in hidden Markov models. Existing algorithms often
suffer from slow computational time or degeneracy. We propose in this paper
a way to improve any of them with a linear complexity in the number of
particles. When iteratively applied to the degenerated Filter-Smoother, this
method leads to an algorithm which turns out to outperform existing linear
particle smoothers for a fixed computational time. Moreover, the associated approximation satisfies a central limit theorem with a close-to-optimal asymptotic variance, which be easily estimated by only one run of the algorithm.
\end{abstract}


\section{Introduction}
\label{sec:intro}
\footnote{This work is supported by the Agence Nationale de la Recherche (ANR, 212, rue de Bercy 75012 Paris) through the 2009-2012 project Big MC}

A \emph{hidden Markov model} (HMM) is a doubly stochastic process where a Markov chain $\{ X_t \}_{t = 0}^\infty$ is only partially observed through a sequence of observations $\{ Y_t \}_{t = 0}^\infty$. More precisely, let $\Xset$ and $\Yset$ be two spaces equipped with countably generated $\sigma$-fields
$\Xsigma$ and $\Ysigma$, respectively, and denote by $M$ a Markovian transition kernel on $(\Xset, \Xsigma)$ and by $G$ a transition kernel from
$(\Xset, \Xsigma)$ to $(\Yset, \Ysigma)$. In our setting, the dynamics of the bivariate process $\{ (X_k, Y_k) \}_{k = 0}^\infty$ follows the Markovian transition kernel
\begin{equation}
\label{eq:JointChainHMM}
P\left[(x,y), \Aset \right] \eqdef M \otimes G [(x,y), \Aset] = \iint M(x, \rmd x') \, G(x', \rmd y') \one_{\Aset}(x', y') \eqsp,
\end{equation}
where $(x,y) \in \Xset \times \Yset$ and $\Aset \in \Xsigma \otimes \Ysigma$.

We assume that there exist nonnegative $\sigma$-finite measures $\lambda$ on $(\Xset, \Xsigma)$ and $\mu$ on $(\Yset, \Ysigma)$ such that for any $x \in \Xset$, $M(x, \cdot)$ and $G(x, \cdot)$ are dominated by $\lambda$ and $\mu$, respectively. This implies the existence of kernel densities
$$
m(x, x') \eqdef \frac{\rmd M(x, \cdot)}{\rmd \lambda}(x') \quad \mbox{and} \quad g(x, y) \eqdef \frac{\rmd G(x, \cdot)}{\rmd \mu}(y) \eqsp.
$$
In what follows, we simply write $\rmd x$ for $\lambda(\rmd x)$.

We are interested here in estimating the expectation of a function of $(X_{0},\ldots,X_{T})$ conditionally on the observations $Y_0,\ldots, Y_T$ using particle smoothing algorithms. Many different implementations of the particle filters and smoothers have been proposed in
the literature with different computational costs; see for example
\cite{delmoral:2004,cappe:moulines:ryden:2005,doucet:johansen:2009}. So far, the existing particle smoothers rely on the so-called {\em Forward-Filter} whose complexity is linear in the number of particles $N$. In its simplest extension, storing the paths of the Forward-Filter allows to approximate the joint smoothing distribution as seen by \cite{kitagawa:1996}. This method known as the {\em Filter-Smoother} unfortunately suffers from a poor representation of the states corresponding to times $t\ll T$. To circumvent this drawback, the {\em FFBS (Forward Filtering Backward Smoothing)} algorithm introduced by \cite{doucet:godsill:andrieu:2000} adds a backward pass to the forward filter at the cost of a quadratic complexity when used for approximating the marginal smoothing distributions. However, \cite{godsill:doucet:west:2004} extended it to the {\em FFBSi (Forward Filtering Backward Simulation)}, an algorithm which can be
implemented with a $\bigo{N}$ computational cost per time step as proposed by \cite{douc:garivier:moulines:olsson:2010} when approximating the whole joint smoothing distribution. If we are interested only in approximations of the marginal smoothing distributions, the {\em Two-Filter smoother} of \cite{briers:doucet:maskell:2010} may also be used as an alternative method. This algorithm originally suffers from a quadratic computational cost but has recently been modified in \cite{fearnhead:wyncoll:tawn:2010} to get a linear one.

Whereas more and more SMC-based smoothing algorithms are linear in the number of particles,
there is a recent surge of interest in mixed strategies (see \cite{andrieu:doucet:holenstein:2010,olsson:ryden:2010} or \cite{chopin:jacob:papaspiliopoulos:2011}) where nice properties of SMC and MCMC algorithms are conjugated to produce better approximations. Whereas these methods are developed mostly in the framework of Bayesian inference for state space models, we focus here on the quality of the approximation of the smoothing distribution associated to a fixed Hidden Markov model. This is a crucial problem to address and the hope is to exhibit the key factors that affects the quality of the estimation. More precisely, fix (once and for all) a set of observations $Y_0,\dots,Y_T$ and try to approximate the law of $X_0,\dots,X_T$ conditionally on the observations with a set of particles $(\xi_0^{i,N},\dots,\xi_T^{i,N})_{i=1}^N$ associated to equal or unequal weights $(\omega_T^{i,N})_{i=1}^N$. For a fixed CPU time, how to build the best population of particles? Should we use mixed strategies? Can we obtain confidence intervals without additional Monte Carlo passes? These are some of the questions we consider in this work. Since $T$ is fixed, the context of this work does not exactly correspond to the one of \cite{gilks:berzuini:2001b} who propose to sequentially alternate SMC stages and MCMC stages as more and more observations are available. Nevertheless, the MCMC step called the Move stage by these authors is now included in the method proposed in this paper to form an efficient algorithm where some directional update of the components extends sequentially the diversity of the population from high values of $t$ to lower values of $t$.  Despite its simplicity, the resulting algorithm turns out to be more than a strong competitor to existing smoothing samplers.

We propose here to improve any consistent particle approximation of the joint smoothing distribution by moving sequentially the particles according to a Metropolis-within-Gibbs iteration. Such algorithm has a linear computational cost and can be applied in particular to the Filter-Smoother to reduce the degeneracy without increasing the complexity. The paper is organized as follows: in Section~\ref{sec:newAlgo}, we describe the algorithm. In Section~\ref{sec:properties}, we show that the limiting variance of the algorithm is reduced in comparison with the original SMC-based population with a multinomial resampling stage. One major characteristic of this algorithm is the fact that, by letting the number of iterations of the Markov chains proportional to $\ln N$, the asymptotic variance is close to optimal and can be estimated using the evolution of only one population of particle paths. Up to our knowledge, this feature is totally new in the smoothing literature. Numerical experiments and comparisons with existing linear smoothers are provided in Section \ref{sec:experiments} for the Linear Gaussian Model (LGM) and the Stochastic Volatility Model (StoVolM).

\section{MH-Improvement of a particle path population}
\label{sec:newAlgo}
Denote for $u\leq s$, $a_{u:s}=(a_u,a_{u+1},\dots,a_s)$ and define the smoothing distribution $\Post{0:T}{T}$ associated to a fixed set of observations $Y_{0:T}=y_{0:T}$ by: for any  $\Aset \in \Xsigma^{\otimes (T+1)}$,
\begin{equation*}
\Post{0:T}{T}(\Aset)\eqdef  \dfrac{\idotsint   \chi(\rmd x_0) g(x_0,y_0) \left[\prod_{i=1}^T m(x_{i-1},x_i)g(x_i,y_i) \right]\one_\Aset(x_{0:T})\rmd x_{1:T}}{\idotsint \chi(\rmd x_0) g(x_0,y_0) \left[\prod_{i=1}^T m(x_{i-1},x_i)g(x_i,y_i) \right]\rmd x_{1:T} } \eqsp,
\end{equation*}
where $\chi$ is a probability measure on $(\Xset,\Xsigma)$.
The distribution $\Post{0:T}{T}$ is thus the law of $X_{0:T}$ conditionally to $Y_{0:T}=y_{0:T}$ when $X_0$ follows the distribution $\chi$. In the sequel,  $\chi$ is assumed to have a density w.r.t. $\lambda(\rmd x)$, density which will be denoted by $\chi$ by abuse of notation: $\chi(\rmd x)=\chi(x) \lambda(\rmd x)$. Then, the density $\post{0:T}{T}$ of the distribution $\Post{0:T}{T}$ with respect to $\prod_{t=0}^T \lambda(\rmd x_t)$ writes
\begin{equation}\label{eq:def-pi}
\post{0:T}{T}(x_{0:T}) \propto \chi( x_0) g(x_0,y_0) \left[\prod_{i=1}^T m(x_{i-1},x_i)g(x_i,y_i) \right] \eqsp.
\end{equation}

As noted in \cite{gilks:berzuini:2001b}, the smoothing density $\post{0:T}{T}$ in \eqref{eq:def-pi} is known up to a normalizing constant so that approximation of this distribution can be perfectly cast into the general framework of the Metropolis-Hastings algorithm. Given that the resulting Markov chain evolves in the path space $\Xset^{T+1}$, the candidate at each iteration should be carefully chosen to keep the acceptance rate away from zero which is a delicate task in high dimensional spaces. Considering this, an appealing approach in the MCMC literature is the Gibbs sampler and more generally the Metropolis-within-Gibbs sampler which proposes to update only one component at a time. One could also choose to update components by blocks but as will be seen in Section~\ref{sec:experiments}, moving only one component at a time is sufficient for our purpose. A key point for exploring the posterior distribution within a reasonable number of iterations is that the algorithm should be well initialized at least for the first components to be updated. We propose here to achieve this by exploiting approximation of $\Post{0:T}{T}$ provided by SMC-based algorithms.

 More precisely, suppose that we already have an approximation of $\Post{0:T}{T}$ through a set of (normalized) weighted particle paths, $(\xi_{0:T}^{i,N},\omega_{0:T}^{i,N})_{i=1}^N$ in the sense that
\begin{equation}\label{eq:first-approx}
\Post{0:T}{T}(h) \approx  \sum_{i=1}^N \omega_{0:T}^{i,N} h(\xi_{0:T}^{i,N})\eqsp, \quad \sum_{i=1}^N \omega_{0:T}^{i,N}=1\eqsp,
\end{equation}
 We intend here to improve this approximation by running $N$ independent Metropolis-within-Gibbs Markov chains $(\xi_{0:T}^{i,N}[k], k \geq 0)$ for  $i \in \{1,\dots,N\}$ starting from each path $\xi_{0:T}^{i,N}$, that is, we set  $\xi_{0:T}^{i,N}[0]=\xi_{0:T}^{i,N}$ for $i \in \{1,\dots,N\}$. The resulting approximation after $K$ iterations of the Markov chains then writes
\begin{equation}\label{eq:approx-post-MWG}
\Post{0:T}{T}(h) \approx  \sum_{i=1}^N \omega_{0:T}^{i,N} h(\xi_{0:T}^{i,N}[K])\eqsp.
\end{equation}
Let us now detail the transition of $(\xi_{0:T}^{i,N}[k],\ k\geq 0)$. For a simpler exposition, we drop here the dependence on $i,N$. Now, consider a family of transition kernel densities $(r_t)_{0 \leq t\leq T}$ such that  $r_0,r_T$ are transition kernel densities on $(\Xset,\Xsigma)$ whereas for $t\in\{1,\dots,T-1\}$, $r_t$ is a transition kernel density on $(\Xset \times \Xset,\Xsigma)$. For $u,v,w,x \in \Xset$, set
\begin{eqnarray}
\alpha_0(v,w;x)&\eqdef&\frac{\chi(x)g(x,y_0)m(x,w)}{\chi(v)g(v,y_0)m(v,w)}\frac{r_0(w;v)} {r_0(w;x)}  \wedge 1\eqsp, \label{eq:def-alpha-0}\\
\alpha_t(u,v,w;x)&\eqdef&\frac{m(u,x)g(x,y_t)m(x,w)}{m(u,v)g(v,y_t)m(v,w)}\frac{r_t(u,w;v)} {r_t(u,w;x)}  \wedge 1\eqsp,\quad 1\leq t \leq T-1\eqsp, \label{eq:def-alpha-t}\\
\alpha_T(u,v;x)&\eqdef&\frac{m(u,x)g(x,y_T)}{m(u,v)g(v,y_t)}\frac{r_T(u;v)} {r_T(u;x)}  \wedge 1 \eqsp. \label{eq:def-alpha-T}
\end{eqnarray}
At time $k$, the new path $\xi_{0:T}[k]$ is obtained by updating backward in time each component $\xi_{t}[k]$ as follows
\begin{enumerate}[(i)]
\item Sample a candidate $X\sim r_t(\xi_{t-1}[k-1],\xi_{t+1}[k],\cdot)$,
\item Accept $\xi_t[k]=X$ with probability $\alpha_t(\xi_{t-1:t}[k-1],\xi_{t+1}[k];X)$,
\item Otherwise, set $\xi_t[k]=\xi_t[k-1]$.
\end{enumerate}
This procedure is valid for $t\in \{1,\dots,T-1\}$; we skip the description of the updates for $\xi_{0}[k]$ and $\xi_{T}[k]$ since they follow the same lines under very slight modifications. The complete pseudo-code version of the Metropolis-Hastings Improved Particle Smoother (MH-IPS)  is given below.
\begin{algorithm}[h]
    \begin{algorithmic}[1]
        \caption{\quad MH-IPS}\label{alg:imprMWG}
        \State \underline{\em Initialization}
        \State Run an SMC-algorithm targeting $\Post{0:T}{T}$ and store $(\xi_{0:T}^{i,N},\omega_{0:T}^{i,N})_{i=1}^N$.
        \State Set: $\forall\ 1\leq i \leq N,\ \xi_{0:T}^{i,N}[0]=\xi_{0:T}^{i,N}$.
        \State \underline{\em $K$ improvement passes}
        \For{$k$ from $1$ to $K$}
        \For{$i$ from $1$ to $N$}
            \State Sample $X \sim r_T(\xi_{T-1}^{i,N}[k-1];\cdot)$,
            \State Accept $\xi_{T}^{i,N}[k]=X$ with probability $\alpha_T(\xi_{T-1:T}^{i,N}[k-1],X)$,
            \State Otherwise, set $\xi_{T}^{i,N}[k]=\xi_{T}^{i,N}[k-1]$.
            \For{$t$ from $T-1$ down to $1$}
                \State Sample $X \sim r_t(\xi_{t-1}^{i,N}[k-1],\xi_{t+1}^{i,N}[k];\cdot)$,
                \State Accept $\xi_{t}^{i,N}[k]=X$ with probability $\alpha_t(\xi_{t-1:t}^{i,N}[k-1],\xi_{t+1}^{i,N}[k],X)$,
                \State Otherwise, set $\xi_{t}^{i,N}[k]=\xi_{t}^{i,N}[k-1]$.
            \EndFor
            \State Sample $X \sim r_0(\xi_{1}^{i,N}[k];\cdot)$,
            \State Accept $\xi_{0}^{i,N}[k]=X$ with probability $\alpha_0(\xi_{0}^{i,N}[k-1],\xi_{1}^{i,N}[k],X)$,
            \State Otherwise, set $\xi_{0}^{i,N}[k]=\xi_{0}^{i,N}[k-1]$.
        \EndFor
        \EndFor
    \end{algorithmic}
\end{algorithm}

Straightforwardly, for any $t \in \{0,\dots,T\}$, $\alpha_t$ is the classical Metropolis-Hastings acceptance rate associated to the proposal kernel $r_t$ and the target distribution $\Post{0:T}{T}$. Due to the specific structure of $\Post{0:T}{T}$ whose density is a product of quantities involving consecutive components, the acceptance ratios in \eqref{eq:def-alpha-0}, \eqref{eq:def-alpha-t} and \eqref{eq:def-alpha-T} do not depend on the path space dimension and are therefore nondegenerated. Of course, it is also possible to update each component from an arbitrary number of neighbors. Nevertheless, in the Gibbs Sampler for which all the acceptance rates are equal to one, the $t$-th component is updated according to the distribution of $X_t$ conditionally on $X_{0:t-1},X_{t+1:T},Y_{0:T}$ which only depends on $X_{t-1},X_{t+1},Y_t$. Such dependence suggests that the candidate in the Metropolis-within-Gibbs algorithm should be proposed according to a distribution which only involves its nearest neighbors.

MH-IPS is based on a first approximation of $\Post{0:T}{T}$ given in \eqref{eq:first-approx} whereas some SMC algorithms like the Filter-Smoother are known to suffer from a poor representation of the states close to $0$ but are accurate for states close to $T$. As a consequence, $(\xi_{t}^{i,N})_{i=1}^N$ for large values of $t$ are well-distributed and this set of particles is then propagated to the poorer ones by updating the components backward in time. In other words, instead of a random-scan procedure where components are updated at random, this determistic-scan Metropolis-Hastings algorithm extends the diversity of the particle paths to the lower values of $t$ at each backward pass. The fact that MH-IPS uses the SMC-based approximation just once and then, keep the $N$ Metropolis-within-Gibbs Markov chains independent from each other implies that the path degeneracy vanishes as the number of iterations increases. Strong empirical evidences of this phenomenon are provided in Section~\ref{sec:experiments}.

A last but striking particularity of MH-IPS when compared to classical MH algorithms is the fact that the  approximation \eqref{eq:approx-post-MWG} only involves the states at iteration $K$ of the $N$ Markov chains instead of using all the history of these Markov chains. Indeed, since only one component is updated at a time, the consecutive paths are highly positively correlated so that including them into \eqref{eq:approx-post-MWG} is detrimental to the quality of the approximation. Another advantage of considering only states at iteration $K$ is that the CLT of the approximation \eqref{eq:approx-post-MWG} which is quite easy to establish when $K \propto \ln N$ includes a very simple and close-to-optimal expression of the asymptotic variance. The estimation of this variance can be performed using the evolution of only one population of sample paths. Therefore, on the contrary to all the smoothing algorithms proposed in the literature so far, confidence intervals can be obtained without additional Monte Carlo passes.

\section{Properties of the algorithm}
\label{sec:properties}
In this section, since the number of observations is fixed, $T$ is dropped for simplicity from the notation. For example, we set $\genTarget=\Post{0:T}{T}$, $\genPart{i}{N}=\xi_{0:T|T}^{i,N}$, $\genWeight{i}{N}=\omega_{0:T|T}^{i,N}$ and so on.

The general procedure induced by MH-IPS can be described as follows.
Let $Q$ be a Markov transition kernel on $(\Xset^{T+1},\Xsigma^{\otimes (T+1)})$ with invariant distribution $\genTarget$. Consider a set of normalized weighted particles $(\genPart{i}{N},\genWeight{i}{N})_{i=1}^N$ and move the particles
 {\em independently} according to the kernel $\moveKernel$. To be specific, define $N$ {\em independent} Markov chains $(\genPart[k]{i}{N},k \geq 0)_{i=1}^N$  such that:
\begin{align}
&\genPart[0]{i}{N} = \genPart{i}{N} \eqsp, \label{eq:init-Markov-chains}\\
&\genPart[k+1]{i}{N} \sim \moveKernel(\genPart[k]{i}{N},\cdot)\eqsp,\quad k \geq 0\eqsp. \label{eq:trans-Markov-chains}
\end{align}
According to \eqref{eq:approx-post-MWG}, $\genTarget h$ is approximated after $k$ iterations of the Markov chains by:
\begin{equation}\label{eq:approx-target-Markov-chains}
\genTarget h \approx \sum_{i=1}^N \genWeight{i}{N} h(\genPart{i}{N}[k]), \quad \sum_{i=1}^N \genWeight{i}{N} = 1\eqsp.
\end{equation}


\subsection{A resampling step in the initialization}
Let us first consider the impact of the weights on the quality of the approximation. A resampling step in the initialization consists in replacing the weighted particles $(\genPart{i}{N},\genWeight{i}{N})_{i=1}^N$ by the unweighted particles  $(\genRePart{i}{N},1/N)_{i=1}^N$ such that some unbiasedness condition is fulfilled. Whereas many resampling strategies have been developed in the literature (\cite{liu:chen:1998}, \cite{kitagawa:1998}, \cite{carpenter:clifford:fearnhead:1999}; see also \cite{douc:cappe:moulines:2005} for a brief review of their different properties), we only focus here on the most simple one, the multinomial resampling:
\begin{enumerate}[(i)]
\item $(\genRePart{j}{N})_{j=1}^N$ are independent conditionally on $(\genPart{i}{N},\genWeight{i}{N})_{i=1}^N$,
\item for all $i,j \in \{1,\dots,N\}$, $\PP{\genRePart{j}{N}=\genPart{i}{N}}=\genWeight{i}{N}$.
\end{enumerate}
A straightforward calculation yields:
$$
\Var{}{\sum_{i=1}^N \genWeight{i}{N} h(\genPart{i}{N})} \leq \Var{}{\sum_{i=1}^N  h(\genRePart{i}{N})/N}\eqsp,
$$
showing that at time $0$, the particle system with equal weights is less efficient than the one with original weights. Despite this, the resampling stage discards particles with small weights and duplicates "informative" particles (with high weights). As in the particle filtering theory, our hope is that the resampling stage increases the number of Markov chains starting from interesting regions with respect to the target distribution.

Denote by  $\norm{\TV}{\cdot}$ the total variation norm: $\norm{\TV}{\mu}\eqdef \sup_{\norminf{f}\leq 1 }|\mu(f)|$ where $\norminf{f}\eqdef \sup_{x \in \Xset} |f(x)|$ and assume that
\begin{hypA} \label{assum:Qk-moins-pi-conv-vers-0}
For any $x \in\Xset^{T+1}$, $\lim_{k \to \infty} \norm{\TV}{\moveKernel^k(x,\cdot)-\genTarget}=0$.
\end{hypA}
 Under this assumption, it is straightforward that for any bounded measurable function $h$, $\sum_{i=1}^N \genWeight{i}{N} h(\genPart{i}{N}[k])$ is asymptotically unbiased whatever the weights are, provided their sum is equal to one. To go further, consider the effect of the weights on the second order approximation. The following proposition shows that as the iterations of the Markov chains goes to infinity, the quadratic error tends to a limit which is minimal when all the weights are equal to $1/N$. This advocates for a particle system with equal weights in the initialization as provided by a resampling step before letting evolve the $N$ Markov chains.
\begin{prop} \label{prop:bestWeights}
Assume \A{\ref{assum:Qk-moins-pi-conv-vers-0}}. Then, for any bounded measurable function $h$,
$$
\lim_{k \to \infty} \mathbb{E}\left[ \left( \sum_{i=1}^N \genWeight{i}{N} h(\genPart[k]{i}{N})-\genTarget h \right)^2\right]=\Var{\genTarget}{h}\ \PE{\sum_{i=1}^N \left(\genWeight{i}{N}\right)^2}
 $$
 where $\Var{\genTarget}{h}=\genTarget h^2-\left(\genTarget h\right)^2$. Moreover,  the previous limit is minimized when all the weights are equal:  $\genWeight{i}{N}=1/N$ for all $i \in\{1,\dots,N\}$.
\end{prop}
\begin{proof}
Proof is given the Appendix.
\end{proof}

As a consequence of this proposition, it is assumed in the sequel that the {\em multinomial resampling stage} has been performed in the initialization, i.e. \eqref{eq:init-Markov-chains}, \eqref{eq:trans-Markov-chains} and \eqref{eq:approx-target-Markov-chains} are replaced by
\begin{align}
&\genRePart[0]{i}{N} = \genRePart{i}{N} \eqsp, \label{eq:init-Markov-chains-Re}\\
&\genRePart[k+1]{i}{N} \sim \moveKernel(\genRePart[k]{i}{N},\cdot)\eqsp,\quad k \geq 0\eqsp, \label{eq:trans-Markov-chains-Re}\\
&\genTarget h \approx \sum_{i=1}^N  h(\genRePart{i}{N}[k])/N\eqsp, \label{eq:approx-target-Markov-chains-Re}
\end{align}
Then, according to Proposition~\ref{prop:bestWeights},
\begin{equation}\label{eq:lim-quadratic-error-resampling}
\lim_{k \to \infty} \mathbb{E}\left[ \left( \sum_{i=1}^N h(\genRePart[k]{i}{N})/N-\genTarget h \right)^2\right]=\Var{\genTarget}{h}/N\eqsp.
\end{equation}
Thus, when $N$ is fixed and $k$ goes to infinity, \eqref{eq:lim-quadratic-error-resampling} shows that the approximation cannot be better than having $N$ independent draws from the distribution $\genTarget$. A natural question is now to properly tune the number of iterations $k$ of the Markov chains to the number $N$ of initial points so that the unweighted particles $(\genPart{i}{N}[k],1/N)_{i=1}^N$ have properties close to iid draws according to $\genTarget$ {\em without} letting $k$ go to infinity.
Before treating this question, let us examine some non-asymptotic result with respect to the approximation.
\subsection{Deviation Inequality}
 Noting that  $(\genRePart{i}{N}[k])_{i=1}^N$ are i.i.d conditionally to $\tfcal_0^N\eqdef \sigma\left\{\genRePart{i}{N}, i \in \{1,\dots,N\}\right\}$ and that $\PE{h(\genRePart{i}{N}[k])|\tfcal_0^N}=\moveKernel^k h(\genRePart{i}{N})$, the conditional Hoeffding inequality directly yields:
\begin{prop} \label{prop:deviation}
For any bounded measurable function $h$, any $k\in \nset$ and any $\epsilon>0$,
\begin{multline}
\PP{\left|\sum_{i=1}^N h(\genRePart[k]{i}{N})/N-\genTarget h\right|>\epsilon} \leq 2 \exp\left(- \frac{N \epsilon^2}{2 \left(\osc{h}\right)^2}\right)\\
+\PP{\left|\sum_{i=1}^N \moveKernel^k h(\genRePart{i}{N})/N-\genTarget h\right|>\epsilon/2}\eqsp,  \label{eq:deviation}
\end{multline}
where $\osc{h}=\sup_{u,v \in \Xset} |h(u) -h(v)|$.
\end{prop}
Nevertheless, when reading the inequality in Proposition \ref{prop:deviation}, the question of knowing whether MH-IPS improves or does not improve the approximation is far from being obvious. We now answer this question in terms of the Central Limit Theorem.
\subsection{Central limit theorem}
MH-IPS is based on a first approximation of $\genTarget h$ by a family of normalized weighted particles $(\genPart{i}{N},\genWeight{i}{N})_{i=1}^N$. For various versions of SMC methods, the asymptotic normality of $(\genPart{i}{N},\genWeight{i}{N})_{i=1}^N$ have already been obtained under different techniques (see for example \cite{delmoral:guionnet:1999}, \cite{kunsch:2000}, \cite{chopin:2004}  or \cite{douc:moulines:2008}). The following proposition now focus on the effect of the multinomial resampling on the central limit theorem: whatever SMC method is chosen, if  $(\genPart{i}{N},\genWeight{i}{N})_{i=1}^N$ are asymptotically normal, then $(\genRePart{i}{N},1/N)_{i=1}^N$ are also asymptotically normal with $\Var{\genTarget}{h}$ as an {\em additional} term in the variance.
\begin{prop} \label{prop:clt-resampling}
Assume that $(\genPart{i}{N},\genWeight{i}{N})_{i=1}^N$ are asymptotically normal, in the sense that for any bounded measurable function $h$, there exists $0<\sigma^2(h)<\infty$ such that
$$
N^{1/2} \left[\sum_{i=1}^N \genWeight{i}{N} h(\genPart{i}{N}) - \genTarget h \right] \dlim \mathcal{N}(0,\sigma^2(h))\eqsp.
$$
Then, for any bounded measurable function $h$,
$$
N^{1/2} \left[\sum_{i=1}^N h(\genRePart{i}{N})/N - \genTarget h \right] \dlim \mathcal{N}(0,\Var{\genTarget}{h}+\sigma^2(h)) \eqsp.
$$
\end{prop}
The proof follows closely the lines of \cite[Theorem 1]{chopin:2004} or \cite[Theorem 4]{douc:moulines:2008} and is omitted for the sake of brevity.

Proposition~\ref{prop:clt-resampling} shows the asymptotic normality of $(\genRePart{i}{N}[k],1/N)_{i=1}^N$ for $k=0$. The Markov chains are then run independently according to the transition kernel $Q$ and we now consider the impact on the approximation given in \eqref{eq:approx-target-Markov-chains-Re} for $k=k_N$. To be specific, the following theorem shows that under the assumption that the kernel $Q$ is $V$-geometrically ergodic, for $k_N\propto\ln N$, the unweighted particles $(\genRePart{i}{N}[k_N],1/N)_{i=1}^N$ are asymptotically normal with a reduced asymptotic variance. Define the following set of assumptions:
\begin{hypA} \label{assum:clt}
There exists a measurable function $V: \Xset^{T+1} \to [1,\infty)$ such that
\begin{enumerate}[(i)]
\item \label{item:first} $\genTarget V<\infty$ and for any $x \in  \Xset$ and any $k\in\nset$,  $\moveKernel^kV(x)<\infty\eqsp,$
\item \label{item:second} there exists $\beta \in (0,1)$ such that for any $h \in \calC_V\eqdef\{h; \norminf{h/V}<\infty\}$ and any $x \in  \Xset$,
$$
|\moveKernel^kh(x)-\genTarget h| \leq \beta^k V(x)\eqsp,
$$
\item \label{item:third}the sequence $\{N^{-1}\sum_{i=1}^N V^2(\genRePart{i}{N})\}_{N \geq 1}$ of random variables  is bounded in probability.
\end{enumerate}
\end{hypA}
\A{\ref{assum:clt}}-\eqref{item:first} ensures that the quantities appearing in \A{\ref{assum:clt}}-{\eqref{item:second}} are well defined.  \A{\ref{assum:clt}}-{\eqref{item:second}} shows that $Q$ is $V$-geometrically ergodic. \A{\ref{assum:clt}}-{\eqref{item:third}} is a weak assumption concerning the initial unweighted particles $(\genRePart{i}{N},1/N)_{i=1}^N$. If for example, $(\genRePart{i}{N},1/N)_{i=1}^N$ is consistent with respect to the function $V^2$ in the sense that $\sum_{i=1}^N V^2(\genRePart{i}{N})/N$ converges in probability to $\genTarget{V^2}$, then \A{\ref{assum:clt}}-{\eqref{item:third}} holds. Condition under which such convergence results hold for possibly unbounded functions may be found for example in \cite{douc:moulines:2008}.
\begin{theorem} \label{thm:clt}
Assume \A{\ref{assum:clt}}. Let $(k_N)_{N \geq 0}$ be a sequence of integers such that
\begin{equation}\label{eq:lim-k-n}
\lim_{N \to \infty} k_N+\ln N/(2\ln \beta)=\infty\eqsp.
\end{equation}
 Then, for any $h$ such that $h^2 \in \calC_V$, the following central limit theorem holds:
\begin{equation*}
 N^{-1/2} \sum_{i=1}^N \left[h(\genRePart[k_N]{i}{N}) - \genTarget h \right] \dlim \mathcal{N}(0,\Var{\genTarget}{h})\eqsp.
\end{equation*}
\end{theorem}
\begin{proof}
Proof is given in the Appendix.
\end{proof}

Theorem~\ref{thm:clt} and Proposition~\ref{prop:clt-resampling} show that $k_N$ iterations of the Markov chains reduce the asymptotic variance when compared to a sample obtained by multinomial resampling of a population issued from any SMC method. The asymptotic variance  $\Var{\genTarget}{h}$ in Theorem~\ref{thm:clt} is close to optimal since it is the same as  for i.i.d. draws with distribution $\genTarget$. Moreover, the expression of $\sigma^2(h)$ in Proposition~\ref{prop:clt-resampling} is usually quite involved and for obtaining confidence intervals, the estimation of the asymptotic variance in Proposition~\ref{prop:clt-resampling} is classically obtained by adding some Monte Carlo passes. This is not at all the case in Theorem~\ref{thm:clt} since estimation of $\Var{\genTarget}{h}$ can be performed directly via $(\genRePart[k_N]{i}{N},1/N)_{i=1}^N$. Finally, by adding typically $k_N=-\ln N /\ln \beta$ iterations of a transition kernel to a SMC-based population of particles, we obtain a sample with a reduced and close-to-optimal variance which can be easily approximated without additional simulations.

The fact that the CLT holds for $k_N \propto \ln N$ suggests that a good approximation of the target distribution may be achieved with only a few number of iterations of the parallel Markov chains. This will be confirmed empirically in the next section.
\section{Experiments}
\label{sec:experiments}
The {\em Filter-smoother} is known to be quite easy to implement and efficient in terms of CPU time, but suffers dramatically from the degeneracy of the ancestors. We now see how only a few iterations of MH-IPS reduce the degeneracy and turn the {\em Filter-smoother} to a strong competitor to the existing smoother algorithms. In the sequel, denote by the \emph{ Metropolis-Hastings Improved Filter-Smoother} (MH-IFS), Algorithm~\ref{alg:imprMWG} initialized with the Filter-Smoother. The performance of this algorithm is now compared to the other linear-in-$N$ particle smoothers (Filter-Smoother, FFBSi, Two-Filter). In order to be as computationally fair as possible, all these algorithms are implemented in the same way as
their common base, the Forward-Filter.

\subsection{Linear Gaussian Model}
\label{subsec:LGM}
We first consider the LGM defined by:
\begin{equation*}
  X_{t+1} = \phi X_t + \sigma_uU_t\eqsp, \quad Y_t = X_t +  \sigma_vV_t\eqsp,
\end{equation*}
where $X_0\sim\mathcal{N}\left(0,\frac{\sigma_u^2}{1-\phi^2}\right)$,
$\left\{U_t\right\}_{t\geq 1}$ and $\left\{V_t\right\}_{t\geq 1}$ are independent sequences of i.i.d.
standard gaussian random variables (independent of $X_1$).
$T+1=101$ observations
were generated using the model with $\phi = 0.9$, $\sigma_u
  = 0.6$ and $\sigma_v = 1$. Furthermore, in this model, the fully-adapted filters are explicitly computable when needed and the Gibbs sampler may be implemented.

The diversity of the particle population at each time step for each algorithm is measured by an estimate of the {\em effective sample size} $N_{\mathrm{eff}}^{\mathrm{algo}}(t)$ as defined in \cite{fearnhead:wyncoll:tawn:2010}. Motivated by the fact that
$\mathbb{E}\left[ \left(\bar X_N-\mu\right)^2/\sigma^2  \right] = 1/N,$
when $X^{(1)},\dots,X^{(N)}$ are \iid~with $\mathbb{E}[X^{(1)}]=\mu$, $\mathbb{V}\mathrm{ar}(X^{(1)}) = \sigma^2$ and $\bar X_N$ is their sample mean, we set
\begin{equation}
\label{eq:Neff}
N_{\mathrm{eff}}^{\mathrm{algo}}(t) \eqdef \mathbb{E}\left[ \left(\dfrac{\post[algo]{t}{T}(\mathrm{Id})-\mu_t}{\sigma_t}\right)^2\right]^{-1}\eqsp,
\end{equation}
where $\mathrm{Id}$ is the identity function on $\mathbb{R}$, $\mu_t$ and $\sigma_t^2$ are the exact mean and variance of $X_t$ conditionally to $Y_{0:T}$ obtained from the Kalman smoother. In some sense, the weighted sample produced by a given algorithm is as accurate at estimating $X_t$ as an "independent" sample of size $N_{\mathrm{eff}}^{\mathrm{algo}}(t)$. The expression of $N_{\mathrm{eff}}^{\mathrm{algo}}(t)$ given in \eqref{eq:Neff} shows that it is inversely proportional to the quadratic error associated to a normalized estimator of ${\mathbb E}(X_t|Y_{0:T})$. To estimate the expectation in \eqref{eq:Neff} we use the mean value from $250$ repetitions of each algorithm with a number of particles chosen such that the computation time of each of them is the same.
\begin{figure}[!h]

\begin{minipage}[b]{1.0\linewidth}
  \centering
  \centerline{\epsfig{figure=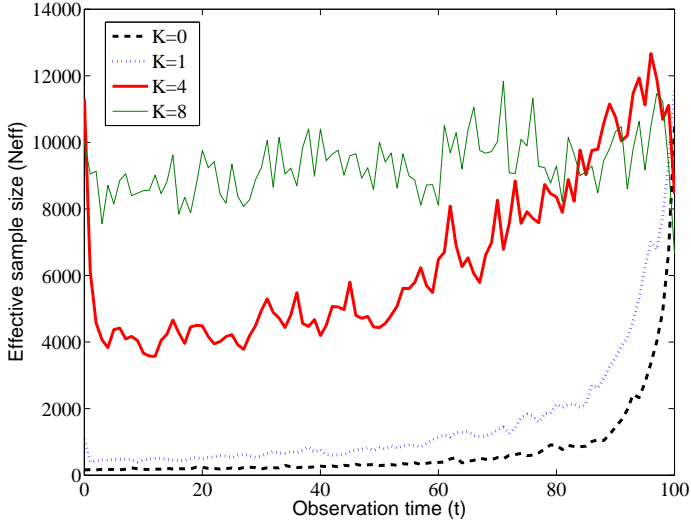,width=10cm}}
  \centerline{(a) Influence of the number of improvements $K$}\medskip
\end{minipage}
\begin{minipage}[b]{1.0\linewidth}
  \centering
  \centerline{\epsfig{figure=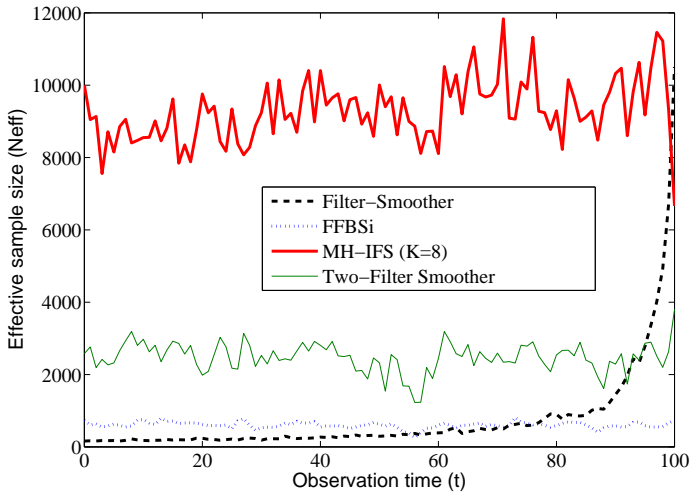,width=10cm}}
  \centerline{(b) Comparison of four linear smoothing algorithms}\medskip
\end{minipage}
\caption{Average effective sample size for each of the 100 time steps of the LGM using different smoothing algorithms for a fixed CPU time.}
\label{fig:LGM_Neff}
\end{figure}

Figure \ref{fig:LGM_Neff}.a shows that when the number of improvements increases, the degeneracy of the particle population for small values of $t$ decreases and for $K=8$ all the time steps have the same diversity.

Figure \ref{fig:LGM_Neff}.b displays the effective sample size of the four linear smoothing algorithms. As expected, the Filter-Smoother is highly degenerated for small values of $t$ as opposed to the other algorithms. Furthermore, the MH-IFS clearly outperforms all others within a fixed computational time. In order to check that this efficiency is not due to the fact that the LGM allows to easily implement the Gibbs sampler, we now turn to a model where a rejection sampling is required.

\subsection{Stochastic Volatility Model}
\label{subsec:SVM}
StoVolM have been introduced in financial time series modeling to capture more realistic features than ARCH/GARCH models (\cite{hull:white:1987}). Despite its apparent simplicity, the following equations do not allow to directly simulate according to $r_t(u,w;\cdot) \propto m(u,\cdot)g(\cdot,y_t)m(\cdot,w)$:
\begin{equation*}
X_{t+1} = \alpha X_t + \sigma U_{t+1}\eqsp, \quad
Y_t = \beta \rme^{\frac{X_t}{2}} V_t\eqsp,
\end{equation*}
where $X_0\sim\mathcal{N}\left(0,\frac{\sigma^2}{1-\alpha^2}\right)$, $U_t$ and $V_t$ are independent standard gaussian random variables. $T+1=101$ observations
were generated using the model with $\alpha = 0.3$, $\sigma= 0.5$ and $\beta = 1$ in order to estimate the effective sample size defined in \eqref{eq:Neff}. The true values of $\mu_t$ and $\sigma_t$ cannot be computed explicitly so they are estimated by running the MH-IFS with $N=650000$.

\subsubsection{Gibbs sampler}
In the StoVolM, the Gibbs sampler requires to sample exactly from
\begin{equation} \label{eq:DirectGibbsStoVolM}
r_t(u,w;x) \propto \exp\left\{-\frac{\e^{-x}}{2\beta^2} y_t^2 - \frac{1+\alpha^2}{2\sigma^2}\left[ x-\left(\frac{\alpha}{1+\alpha^2}(u+w) - \frac{\sigma^2/2}{1+\alpha^2}\right)\right]^2 \right\} \eqsp,
\end{equation}
for $1 \leq t \leq T-1$ (the cases $t=0$ and $t=T$ are dealt with in a similar way) which does not correspond to a classical distribution. However, we propose here to implement a rejection sampling. The first idea is to sample the proposal candidate $X=x$ according to the \emph{a priori} distribution of $X_t$ conditionally to $X_{t-1}=u$ and $X_{t+1}=w$. The corresponding ratio of acceptance is then given by $(|y_t|/\beta) \exp\left\{-(x-1)/2 - \e^{-x}y_t^2/(2\beta^2)\right\}$ and will obviously lead to poor results for small values of $y_t$. To counterbalance the effect of $y_t$ in the acceptance rate, the proposal distribution should also take the value of $y_t$ into account; we then rewrite \eqref{eq:DirectGibbsStoVolM} for any $\gamma_t \geq 0$ (possibly depending on $y_t$):
\begin{equation} \label{eq:GibbsStoVolMwithGamma}
r_t(u,w;x) \propto \e^{-\frac{\gamma_t}{2}x-\frac{\e^{-x}}{2\beta^2} y_t^2} \times \exp\left\{ - \frac{1+\alpha^2}{2\sigma^2}\left[ x-\left(\frac{\alpha}{1+\alpha^2}(u+w) - \frac{\sigma^2/2}{1+\alpha^2}(1-\gamma_t)\right)\right]^2 \right\} \eqsp,
\end{equation}
which suggests to propose $x$ according to $\mathcal{N}\left( \frac{\alpha}{1+\alpha^2}(u+w) - \frac{\sigma^2/2}{1+\alpha^2}(1-\gamma_t), \frac{\sigma^2}{1+\alpha^2}\right)$ and to accept it with a probability given by:
\begin{equation} \label{eq:GibbsStoVolMRatio}
\left(\frac{|y_t|}{\gamma_t^{1/2}\beta}\right)^{\gamma_t} \exp\left\{-\frac{\gamma_t}{2}(x-1) - \frac{\e^{-x}}{2\beta^2}y_t^2\right\} \eqsp.
\end{equation}
An optimal choice for $\gamma_t$ would consist in maximizing the smoothed expectation of \eqref{eq:GibbsStoVolMRatio} but this quantity is intractable. An intuitive choice for $\gamma_t$ is then:
\begin{equation} \label{eq:GammaStoVolM}
\gamma_t = \begin{cases}
(|y_t|/\beta)^2\eqsp, & \mathrm{if }\ |y_t| \leq \beta\eqsp, \\
|y_t|/\beta\eqsp, & \mathrm{if }\ |y_t| > \beta\eqsp.
\end{cases}
\end{equation}
Indeed, for small values of $y_t$, \eqref{eq:GibbsStoVolMRatio} is then close to one and for bigger values, the exponential becomes very small but the first term remains non-neglectable.

\begin{figure}[!h]

\begin{minipage}[b]{1.0\linewidth}
  \centering
  \centerline{\epsfig{figure=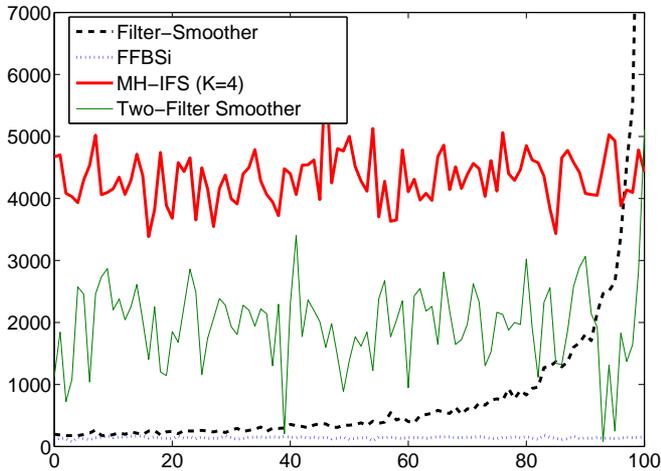,width=10cm}}
\end{minipage}
\caption{Average effective sample size for each of the 100 time steps of the StoVolM using different smoothing algorithms for a fixed CPU time.}
\label{fig:SVM_Neff}
\end{figure}

The Improved Filter-Smoother used to generate Figure \ref{fig:SVM_Neff} performs simulations using the Gibbs sampler with the previous rejection sampling. We can see that this algorithm still leads to better results than the other ones within an equivalent computational time.

In many instances (for example Expectation-Maximization algorithm, score computation), it is necessary to estimate smoothed additive functionals such as $\Post{0:T}{T}(H)$ where for all $x_{0:T} \in \Xset^{T+1}$, $H(x_{0:T}) = \sum_{t=0}^T x_t$. In order to assess the smoothing algorithms on this matter, $T+1=1001$ observations were generated. As seen before, the computational cost of the MH-IFS is linear in $N$ which is verified by numerical experiments in Figure \ref{fig:LGM_CPU}.

\begin{figure}[!h]

\begin{minipage}[b]{1.0\linewidth}
  \centering
  \centerline{\epsfig{figure=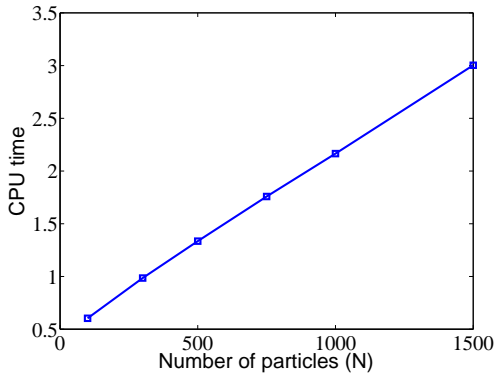,width=7cm}}
\end{minipage}
\caption{Average CPU time for computing a smoothed additive functional with the MH-IFS as a function of the number of particles.}
\label{fig:LGM_CPU}
\end{figure}

\begin{figure}[!h]

\begin{minipage}[b]{1.0\linewidth}
  \centering
  \centerline{\epsfig{figure=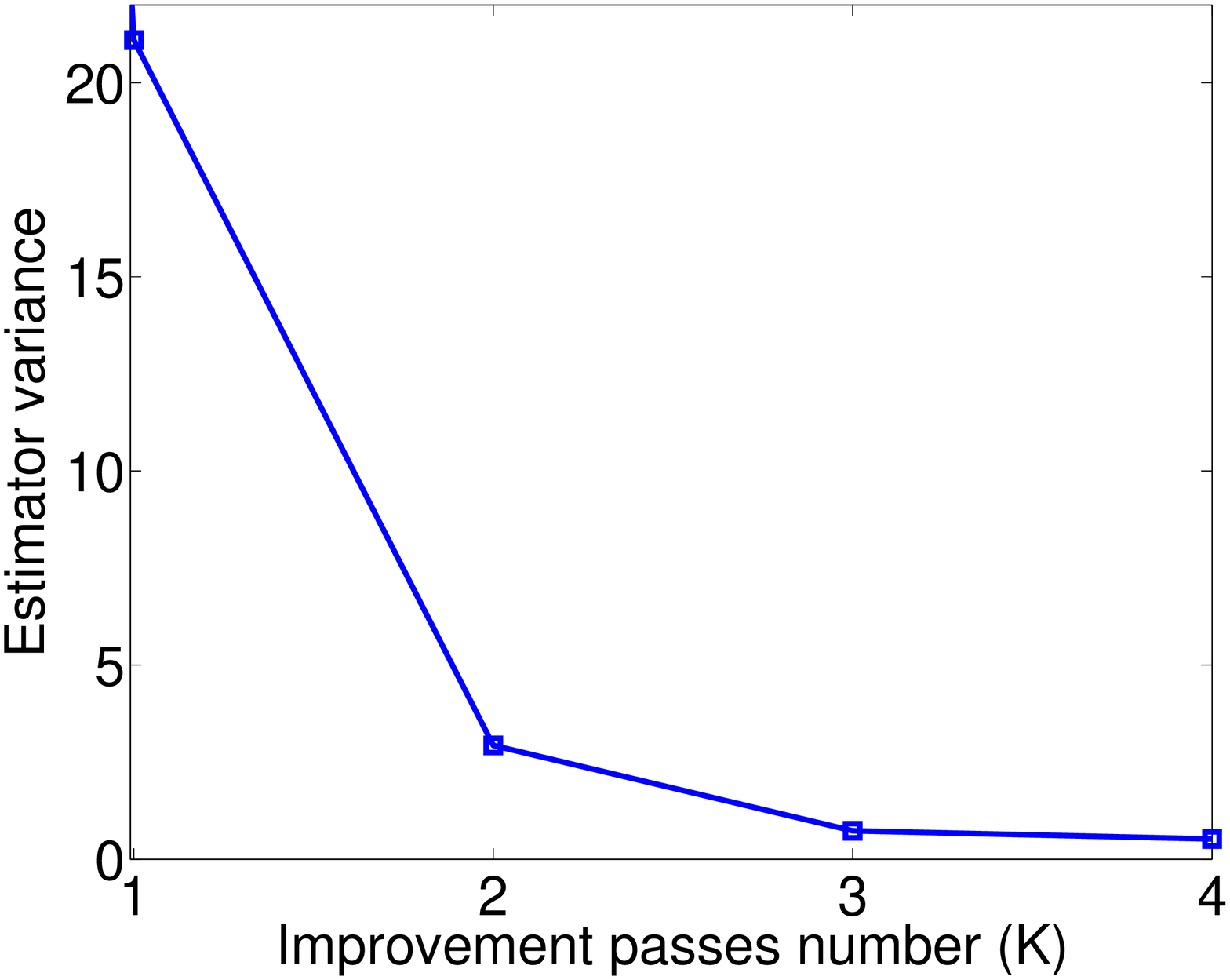,width=6.5cm}}
  \centerline{(a) Variance of the Improved Filter-Smoother}
  \centerline{according to the number of improvement passes $K$}\medskip
\end{minipage}
\begin{minipage}[b]{1.0\linewidth}
  \centering
  \centerline{\epsfig{figure=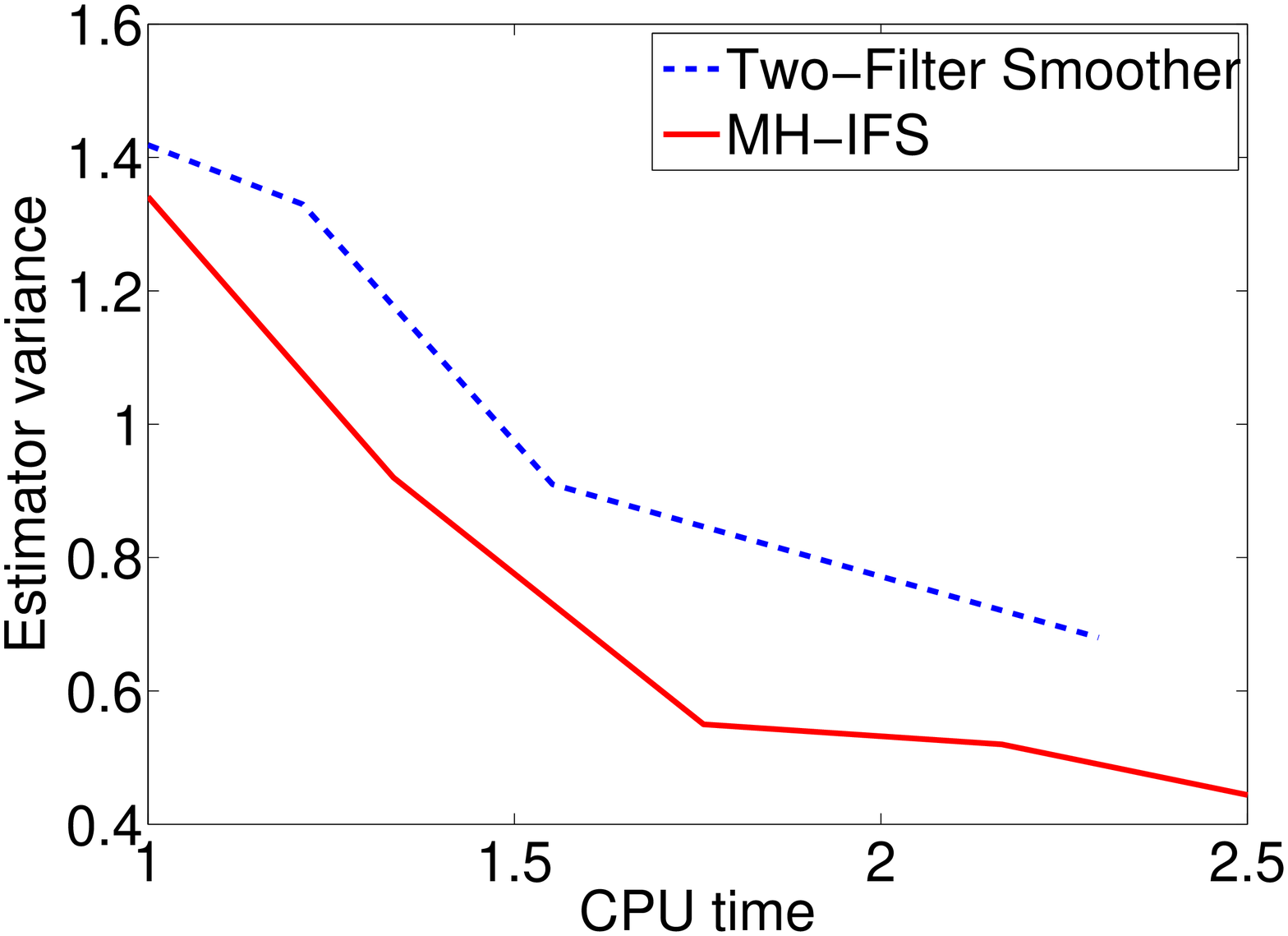,width=6.5cm}}
  \centerline{(b) Variance of the Two-Filter Smoother and the Improved }
  \centerline{Filter-Smoother according to the CPU time}\medskip
\end{minipage}
\begin{minipage}[b]{1.0\linewidth}
  \centering
  \centerline{\epsfig{figure=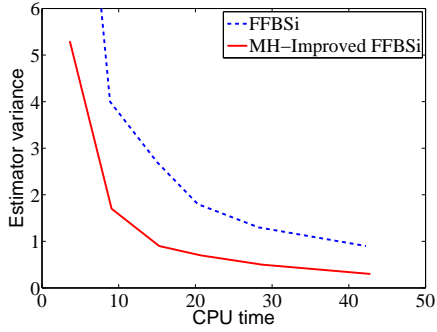,width=6.5cm}}
  \centerline{(c) Variance of the FFBSi and its improved version}
  \centerline{according to the CPU time}\medskip
\end{minipage}
\caption{Variance of different smoothed additive functional particle estimators in the StoVolM.}
\label{fig:SVM_Var}
\end{figure}
Figure \ref{fig:SVM_Var}.a shows that the variance vanishes quickly with the number of improvement passes and only $4$ iterations of the Markov chains are sufficient to get an efficient estimator. Then, the variances displayed in Figure \ref{fig:SVM_Var}.b allow again to draw the conclusion that for a fixed CPU time, the MH-IFS is more efficient than the Two-Filter. Finally, one improvement pass has been applied to the particle paths given by the FFBSi. The variance reduction is again significant as shown in Figure \ref{fig:SVM_Var}.c.

\subsubsection{Metropolis-within-Gibbs and confidence interval}
In order to assess Algorithm \ref{alg:imprMWG} in the case where the Gibbs sampler could not be implemented, we now turn to the Metropolis-within-Gibbs sampler which is implemented by using again the proposal distribution:
\begin{equation*}
r_t(u,w;\cdot) \sim \mathcal{N}\left( \frac{\alpha}{1+\alpha^2}(u+w) - \frac{\sigma^2/2}{1+\alpha^2}(1-\gamma_t), \frac{\sigma^2}{1+\alpha^2}\right) \eqsp,
\end{equation*}
where $\gamma_t$ is defined in \eqref{eq:GammaStoVolM}, and the associated acceptance rate is now given by:
\begin{equation*}
\alpha_t(u,v,w;x) = \exp\left\{-\frac{\gamma_t}{2}(x-v) - \frac{\e^{-x}-\e^{-v}}{2\beta^2}y_t^2\right\} \wedge 1 \eqsp.
\end{equation*}

Figure \ref{fig:GibbsVsMwG} compares the empirical variance of the Gibbs and Metropolis-within-Gibbs samplers of the smoothed additive functional conditionally to the $T+1=1001$ observations used previously. The efficiency of both algorithms is equivalent, showing that Algorithm \ref{alg:imprMWG} remains a great performer even when exact \emph{a posteriori} simulation is not possible.

\begin{figure}[h]

\begin{minipage}[b]{1.0\linewidth}
  \centering
  \centerline{\epsfig{figure=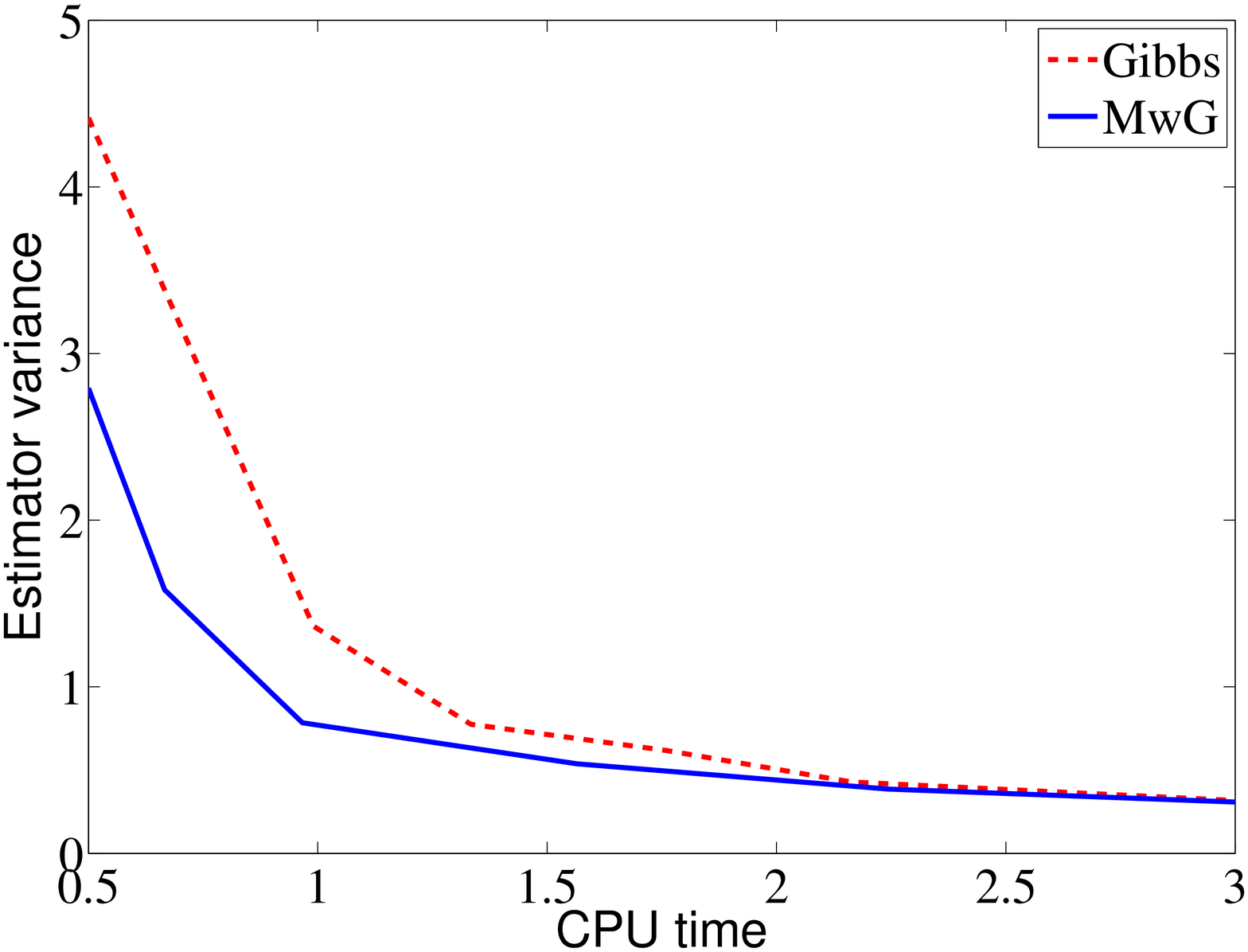,width=8cm}}
\end{minipage}
\caption{Variance of the Gibbs and Metropolis-within-Gibbs samplers according to the CPU time.}
\label{fig:GibbsVsMwG}
\end{figure}

Finally, Theorem \ref{thm:clt} is assessed in Figure \ref{fig:clt}. The empirical variance of the estimator given by Algorithm \ref{alg:imprMWG} run with $K_N \propto \ln N$ has been computed over $250$ runs using the Gibbs and the Metropolis-within-Gibbs samplers for different number of particles $N$ and compared to the asymptotic variance $\Var{\genTarget}{h}/N$ estimated through only one population of particles. The results show that it is possible in practice to get a confidence interval for the approximation with only one run of Algorithm \ref{alg:imprMWG} of complexity $\mathcal{O}(N \ln N)$.

\begin{figure}[!h]

\begin{minipage}[b]{1.0\linewidth}
  \centering
  \centerline{\epsfig{figure=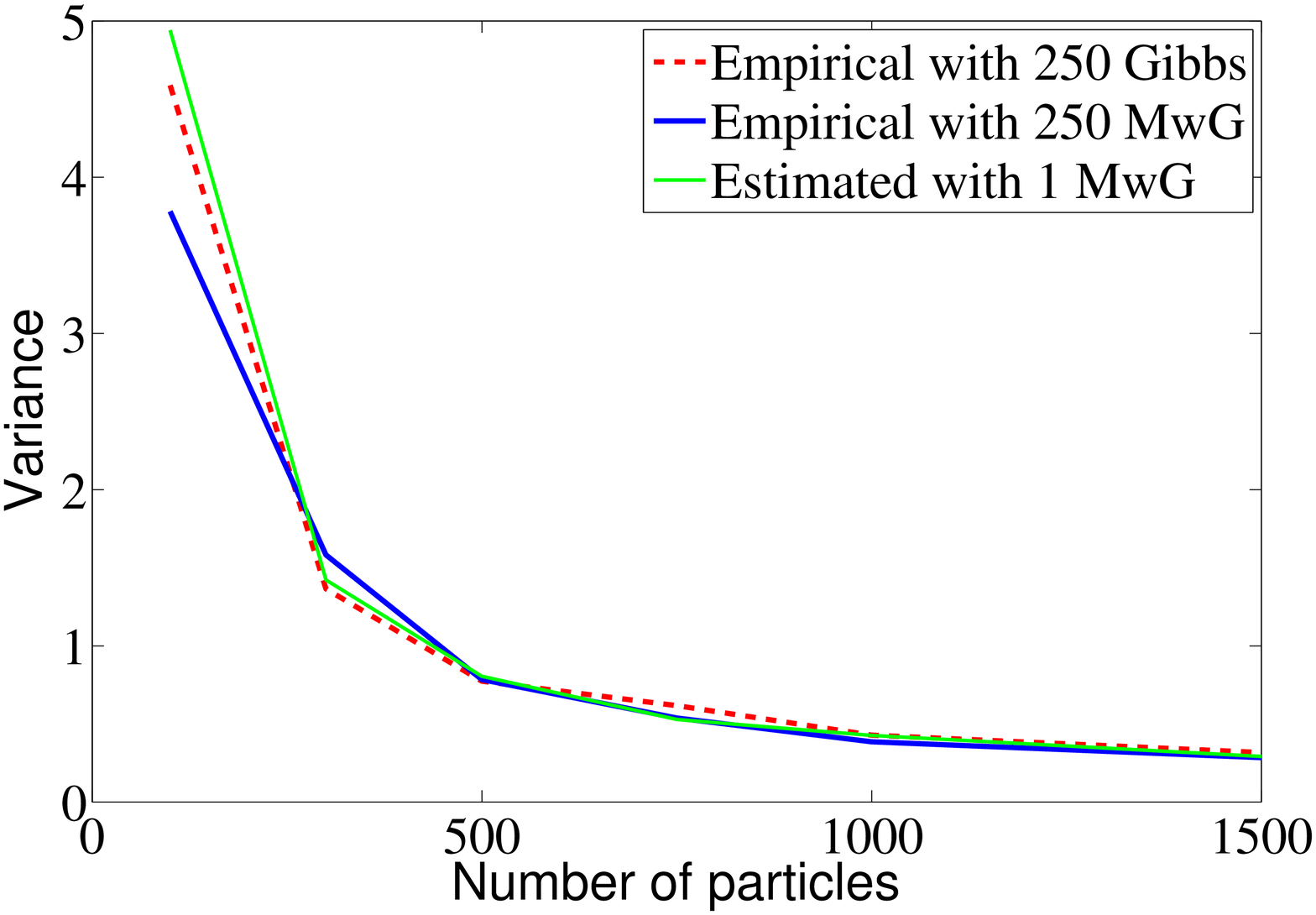,width=8cm}}
\end{minipage}
\caption{Algorithm \ref{alg:imprMWG} variance according to the number of observations.}
\label{fig:clt}
\end{figure}

\section{Conclusion}
At first sight, one could fear that the MH-IPS is too slow since the updates concern only one component at a time. The various comparisons performed for a {\em fixed} CPU time in the previous section show that this {\em is not} the case at all. Roughly speaking, a backward pass in the MH-IPS proposes to sequentially modify each component of the $N$ parallel Markov chains. This can be seen as one run of $N$ particles through $T+1$ observations which is computationally equivalent to one pass of the bootstrap filter. By empirical evidences, we have seen that only a few backward passes ($K=4$ or $8$ in the examples) of the MH-IFS sweep out the degeneracy of the ancestors by extending backward in time the diversity of the particles.

This method is linear in $N$ and outperforms other existing algorithms as the FFBSi or the Two-Filter within a fixed CPU time. These performance results may be explained by the fact that in the FFBSi algorithm, the points are sampled in the forward pass once and for all; the backward pass in the FFBSi only modifies the weights of the particles without moving them. On the contrary, the MH-IPS allows in the backward pass to move the particles and thus to explore interesting regions of the posterior distribution. In the Two-Filter sampler, two populations (the "forward" population and the "backward" population) evolve {\em independently}. At time $t$, a particle is sampled after choosing a couple of particles at time $t-1$ and $t+1$. The two components of these couples belong to independent populations and it is likely that even if their weights are respectively high, associating these independent particles could be detrimental to the approximation. On the contrary, in the MH-IPS, even if the Markov chains are independent, the proposed modification of the component is sampled with respect to its two neighbors which both belong to the {\em same} Markov chain.
 Note that we did not compare this algorithm to the Population Monte Carlo by Markov chains (PMCMC) samplers introduced by \cite{andrieu:doucet:holenstein:2010} since the framework here is not the Bayesian inference of parameterized Markov chains.

 Another major advantage here is the fact that a CLT can be obtained with a very simple asymptotic variance which can be estimated with only one run of the Algorithm and a complexity in $\mathcal{O}(N \ln N)$. This is totally new in comparison to all the smoothing algorithms proposed in the literature so far, where the asymptotic variances are usually particularly involved. Thus, for a fixed CPU time and only one run, this algorithm is able to produce both approximations of the smoothing distributions and confidence intervals.

Finally, we only focus here on the MH-IFS since it is efficient enough for our purpose. Of course, many other variants with different SMC-based approximations in the initialization step may be performed. In the context of the paper, the MH-IPS only uses the SMC-based approximation once before starting independent MCMC Markov chains. The empirical performances of this algorithm, namely with respect to the diversity of the population and the precision of the approximation, seem to us convincing enough to let the Markov chains evolve independently without trying to interact them again. Of course, as previously noted in \cite{gilks:berzuini:2001b}, in some different contexts, where for example, the observations are available sequentially whereas approximations of the smoothing distributions are needed at each time, some variants with SMC steps mixed with MCMC steps can also be elaborated. Nevertheless, in the framework of this paper, the number $T$ of the observations is fixed and we only focus here on how the independent MCMC steps drastically improve the first approximation obtained by SMC algorithms. In this context, there is no need to interact again the Markov chains; this allows to keep the diversity of the population while approximations and confidence intervals are obtained without effort.

\appendix
\section*{Appendix}
\section{Proof of Proposition~\ref{prop:bestWeights}}
For all $k \geq 0$, the bias plus variance decomposition writes
\begin{align*}
&\mathbb{E}\left[ \left( \sum_{i=1}^N \genWeight{i}{N} h(\genPart[k]{i}{N})-\genTarget h \right)^2\right] \\
&\quad =\left\{ \mathbb{E}\left[\sum_{i=1}^N \genWeight{i}{N} h(\genPart[k]{i}{N})\right] - \genTarget h\right\}^2 +\var{\sum_{i=1}^N \genWeight{i}{N} h(\genPart[k]{i}{N})}
\end{align*}
\begin{multline}
\quad = \left\{ \mathbb{E}\left[\sum_{i=1}^N \genWeight{i}{N} h(\genPart[k]{i}{N})\right] - \genTarget h\right\}^2 +\var{ \mathbb{E}\left[\sum_{i=1}^N \genWeight{i}{N} h(\genPart[k]{i}{N})\middle| \mathcal{F}_0^N \right]}  \\
+ \mathbb{E}\left[\var[\mathcal{F}_0^N]{ \sum_{i=1}^N \genWeight{i}{N} h(\genPart[k]{i}{N})}\right] \eqsp,  \label{eq:quadraticErrorDecomposition}
\end{multline}
where $\mathcal{F}_0^N=\sigma\left\{\genPart{i}{N},\genWeight{i}{N},\ i \in\{1,\dots,N\} \right\}$. Now, by definition of $\genPart[k]{i}{N},\ i \in \{1,\ldots,N\}$,
\begin{equation*}
\mathbb{E}\left[\sum_{i=1}^N \genWeight{i}{N} h(\genPart[k]{i}{N})\middle| \mathcal{F}_0^N \right] = \sum_{i=1}^N \genWeight{i}{N} \moveKernel^kh(\genPart{i}{N})\eqsp,
\end{equation*}
and the first term of the RHS of \eqref{eq:quadraticErrorDecomposition} is bounded by
\begin{equation*}
\left| \mathbb{E}\left[\sum_{i=1}^N \genWeight{i}{N} h(\genPart[k]{i}{N})\right] - \genTarget h\right| \leq \mathbb{E}\left[\sum_{i=1}^N \genWeight{i}{N} \left| \moveKernel^kh(\genPart{i}{N}) - \genTarget h \right| \right] \eqsp.
\end{equation*}
The RHS goes to $0$ as $k$ tends to infinity by the Lebesgue convergence theorem since $h$ is bounded. The same argument holds to handle the second term of the RHS of \eqref{eq:quadraticErrorDecomposition}:
\begin{multline*}
\lim_{k \rightarrow \infty} \var{ \mathbb{E}\left[\sum_{i=1}^N \genWeight{i}{N} h(\genPart[k]{i}{N})\middle| \mathcal{F}_0^N \right]} = \lim_{k \rightarrow \infty} \var{ \sum_{i=1}^N \genWeight{i}{N} \moveKernel^kh(\genPart{i}{N})} \\
= \var{ \sum_{i=1}^N \genWeight{i}{N} \genTarget h} = \var{ \genTarget h} = 0\eqsp.
\end{multline*}
Finally, conditionally to $\mathcal{F}_0^N$, the random variables $(\genPart[k]{i}{N})_{i=1}^N$ are independent and
\begin{multline*}
\var[\mathcal{F}_0^N]{ \sum_{i=1}^N \genWeight{i}{N} h(\genPart[k]{i}{N})} = \sum_{i=1}^N (\genWeight{i}{N})^2 \var[\mathcal{F}_0^N]{h(\genPart[k]{i}{N})} \\
= \sum_{i=1}^N (\genWeight{i}{N})^2 \left[ \moveKernel^kh^2(\genPart{i}{N}) - \left(\moveKernel^kh(\genPart{i}{N})\right)^2\right] \eqsp,
\end{multline*}
leading to
\begin{equation*}
\lim_{k \rightarrow \infty} \mathbb{E}\left[\var[\mathcal{F}_0^N]{ \sum_{i=1}^N \genWeight{i}{N} h(\genPart[k]{i}{N})}\right] = \left[\genTarget h^2 - (\genTarget h)^2\right] \mathbb{E}\left[\sum_{i=1}^N (\genWeight{i}{N})^2\right]\eqsp.
\end{equation*}
This shows the first part of the proposition. Now, by the Cauchy-Schwartz inequality:
\begin{equation*}
1=\sum_{i=1}^N \genWeight{i}{N} \leq \left(\sum_{i=1}^N (\genWeight{i}{N})^2\right)^{1/2} N^{1/2}\eqsp,
\end{equation*}
i.e. $\sum_{i=1}^N (\genWeight{i}{N})^2 \geq 1/N$ with equality only for $\genWeight{i}{N} = 1/N$ for all $i$. The proof is completed.

\section{Proof of Theorem~\ref{thm:clt}}
Let $\gamma_N=k_N+\ln N /(2 \ln \beta)$. Under the assumptions of Theorem~\ref{thm:clt}, $\lim_{N \to \infty} \gamma_N=\infty$. Now, write
\begin{multline}\label{eq:cltDecompo}
 N^{-1/2} \sum_{i=1}^N \left[h(\genRePart[k_N]{i}{N}) - \genTarget h \right] = N^{-1/2} \sum_{i=1}^N \left[\moveKernel^{k_N}h(\genRePart{i}{N}) - \genTarget h \right]\\
 + N^{-1/2} \sum_{i=1}^N \left[h(\genRePart[k_N]{i}{N}) - \moveKernel^{k_N}h(\genRePart{i}{N}) \right]\eqsp.
\end{multline}
Since $V\geq 1$, \A{\ref{assum:clt}}-\eqref{item:third} implies that $\{N^{-1}\sum_{i=1}^N V(\genRePart{i}{N})\}_{N \geq 1}$ is bounded in probability. Combining this with
\begin{equation*}
\left|N^{-1/2} \sum_{i=1}^N \left[\moveKernel^{k_N}h(\genRePart{i}{N}) - \genTarget h \right]\right| \leq N^{-1/2}\beta^{k_N}\sum_{i=1}^N V(\genRePart{i}{N}) = \beta^{\gamma_N} \times N^{-1}\sum_{i=1}^N V(\genRePart{i}{N})\eqsp,
\end{equation*}
shows that the first term of the RHS of \eqref{eq:cltDecompo} converges in probability to $0$.
Now, the second term of the RHS of \eqref{eq:cltDecompo} writes
$$
N^{-1/2} \sum_{i=1}^N \left[h(\genRePart[k_N]{i}{N}) - \moveKernel^{k_N}h(\genRePart{i}{N}) \right]=\sum_{i=1}^N \left\{ U_{N,i}-\PE[\fcal_{N,i-1}]{U_{N,i} }\right\}\eqsp,
$$
where
\begin{align*}
&U_{N,i} = N^{-1/2}h\left(\genRePart[k_N]{i}{N}\right)\eqsp, \\
& \mathcal{F}_{N,i} = \sigma\left\{ \genRePart{\ell}{N},\genRePart{j}{N}[k_N], (\ell,j) \in \{1,\dots,i\}^2  \right\} \eqsp.
\end{align*}
To apply \cite[Theorem A3]{douc:moulines:2008} with $M_N=N$ and $\sigma^2 = \Var{\genTarget}{h}$, we need to check that
\begin{align}
& \sum_{i=1}^N \var[\mathcal{F}_{N,i-1}]{U_{N,i}} \plim \sigma^2 \eqsp, \label{eq:conv-varCond-sigma}\\
& \sum_{i=1}^N \PE[\mathcal{F}_{N,i-1}]{U_{N,i}^2 \one_{\{|U_{N,i}|\geq \varepsilon\}}} \plim 0\ ,  \quad \mbox{for any }\epsilon>0\eqsp. \label{eq:conv-tightness-zero}
 \end{align}
We start with \eqref{eq:conv-varCond-sigma}. Write
\begin{multline}\label{eq:triangleVar}
\left| \sum_{i=1}^N \var[\mathcal{F}_{N,i-1}]{U_{N,i}} - \sigma^2 \right| \\
\leq N^{-1} \sum_{j=1}^N  \left| \moveKernel^{k_N} h^2(\genRePart{j}{N})-\genTarget h^2 \right| + N^{-1} \sum_{j=1}^N  \left| \left[\moveKernel^{k_N} h(\genRePart{j}{N})\right]^2-(\genTarget h)^2 \right|\eqsp.
\end{multline}
As $h^2 \in \calC_{V}$, the first term of the RHS is upper-bounded by
$$\beta^{k_N} \times N^{-1}\sum_{i=1}^N V(\genRePart{i}{N})\eqsp,$$
which converges in probability to $0$. Now, note that the functions $h^2$ and $V$ are in $\calC_V$ and $|h|\leq \max(h^2,1)\leq \max(h^2,V)$ so that $h\in \calC_V$. By applying $|a^2-b^2| \leq |a-b|^2 + 2|b||a-b|$, the second term of \eqref{eq:triangleVar} is then upper-bounded by
$$\beta^{2k_N} \times N^{-1}\sum_{i=1}^N \left[V(\genRePart{i}{N})\right]^2 + 2|\genTarget h| \beta^{k_N} \times N^{-1}\sum_{i=1}^N V(\genRePart{i}{N})\eqsp,$$
which again converges in probability to $0$. This proves \eqref{eq:conv-varCond-sigma}. Now, let $\varepsilon > 0$,
\begin{align}
&\sum_{i=1}^N \mathbb{E}\left[U_{N,i}^2 \one_{\{|U_{N,i}|\geq \varepsilon\}}\middle| \mathcal{F}_{N,i-1} \right] \nonumber\\
&\quad \leq \genTarget \left[h^2\one_{\{h^2 \geq \varepsilon^2 N \}}\right]
+ N^{-1}\sum_{i=1}^N \left| \moveKernel^{k_N} \left[h^2(\genRePart{i}{N})\one_{\{h^2(\genRePart{i}{N}) \geq \varepsilon^2 N \}}\right]-\genTarget \left[h^2\one_{\{h^2 \geq \varepsilon^2 N \}}\right] \right| \nonumber\\
&\quad \leq \genTarget \left[h^2\one_{\{h^2 \geq \varepsilon^2 N \}}\right] + \beta^{k_N}\times N^{-1}\sum_{i=1}^N V(\genRePart{i}{N}) \eqsp, \label{eq:ineg-tightness-residus}
\end{align}
where $h^2\one_{\{h^2\geq \varepsilon^2 N \}} \in\calC_V$. Since $h^2 \in \calC_{V}$,\A{\ref{assum:clt}}-\eqref{item:first} implies that $\genTarget h^2 <\infty$. Then, the RHS of
\eqref{eq:ineg-tightness-residus} converges in probability to $0$, showing \eqref{eq:conv-tightness-zero}. The proof is completed.

\bibliographystyle{chicago}
\bibliography{motherofallbibs}

\begin{thebibliography}{}

\bibitem[\protect\citeauthoryear{Andrieu, Doucet, and Holenstein}{Andrieu
  et~al.}{2010}]{andrieu:doucet:holenstein:2010}
Andrieu, C., A.~Doucet, and R.~Holenstein ({2010}).
\newblock Particle markov chain monte carlo methods.
\newblock {\em J. Roy. Statist. Soc. B\/}~{\em {72}\/}({Part 3}), {269--342}.

\bibitem[\protect\citeauthoryear{Briers, Doucet, and Maskell}{Briers
  et~al.}{2010}]{briers:doucet:maskell:2010}
Briers, M., A.~Doucet, and S.~Maskell (2010).
\newblock Smoothing algorithms for state-space models.
\newblock {\em {A}nnals {I}nstitute {S}tatistical {M}athematics\/}~{\em
  62\/}(1), 61--89.

\bibitem[\protect\citeauthoryear{Capp\'{e}, Moulines, and Ryd\'{e}n}{Capp\'{e}
  et~al.}{2005}]{cappe:moulines:ryden:2005}
Capp\'{e}, O., E.~Moulines, and T.~Ryd\'{e}n (2005).
\newblock {\em Inference in Hidden {M}arkov Models}.
\newblock Springer.

\bibitem[\protect\citeauthoryear{Carpenter, Clifford, and Fearnhead}{Carpenter
  et~al.}{1999}]{carpenter:clifford:fearnhead:1999}
Carpenter, J., P.~Clifford, and P.~Fearnhead (1999).
\newblock An improved particle filter for non-linear problems.
\newblock {\em IEE Proc., Radar Sonar Navigation\/}~{\em 146}, 2--7.

\bibitem[\protect\citeauthoryear{Chopin}{Chopin}{2004}]{chopin:2004}
Chopin, N. (2004).
\newblock Central limit theorem for sequential {M}onte {C}arlo methods and its
  application to {B}ayesian inference.
\newblock {\em Ann. Statist.\/}~{\em 32\/}(6), 2385--2411.

\bibitem[\protect\citeauthoryear{Chopin, Jacob, and Papaspiliopoulos}{Chopin
  et~al.}{2011}]{chopin:jacob:papaspiliopoulos:2011}
Chopin, N., P.~Jacob, and O.~Papaspiliopoulos (2011).
\newblock $smc^2$: A sequential monte carlo algorithm with particle markov
  chain monte carlo updates.
\newblock Preprint, arXiv:1011.1528v2.

\bibitem[\protect\citeauthoryear{{Del Moral}}{{Del
  Moral}}{2004}]{delmoral:2004}
{Del Moral}, P. (2004).
\newblock {\em {F}eynman-Kac {F}ormulae. {G}enealogical and Interacting
  Particle Systems with Applications}.
\newblock Springer.

\bibitem[\protect\citeauthoryear{Del~Moral and Guionnet}{Del~Moral and
  Guionnet}{1999}]{delmoral:guionnet:1999}
Del~Moral, P. and A.~Guionnet (1999).
\newblock Central limit theorem for nonlinear filtering and interacting
  particle systems.
\newblock {\em Ann. Appl. Probab.\/}~{\em 9\/}(2), 275--297.

\bibitem[\protect\citeauthoryear{Douc, Capp\'{e}, and Moulines}{Douc
  et~al.}{2005}]{douc:cappe:moulines:2005}
Douc, R., O.~Capp\'{e}, and E.~Moulines (2005, September).
\newblock Comparison of resampling schemes for particle filtering.
\newblock In {\em 4th International Symposium on Image and Signal Processing
  and Analysis (ISPA)}, Zagreb, Croatia.
\newblock arXiv: cs.CE/0507025.

\bibitem[\protect\citeauthoryear{Douc, Garivier, Moulines, and Olsson}{Douc
  et~al.}{2010}]{douc:garivier:moulines:olsson:2010}
Douc, R., A.~Garivier, E.~Moulines, and J.~Olsson (2010, 4).
\newblock Sequential {M}onte {C}arlo smoothing for general state space hidden
  {M}arkov models.
\newblock {\em To appear in Ann. Appl. Probab.\/}.

\bibitem[\protect\citeauthoryear{Douc and Moulines}{Douc and
  Moulines}{2008}]{douc:moulines:2008}
Douc, R. and E.~Moulines (2008).
\newblock Limit theorems for weighted samples with applications to sequential
  {M}onte {C}arlo methods.
\newblock {\em Ann. Statist.\/}~{\em 36\/}(5), 2344--2376.

\bibitem[\protect\citeauthoryear{Doucet, Godsill, and Andrieu}{Doucet
  et~al.}{2000}]{doucet:godsill:andrieu:2000}
Doucet, A., S.~Godsill, and C.~Andrieu (2000).
\newblock On sequential {M}onte-{C}arlo sampling methods for {B}ayesian
  filtering.
\newblock {\em Stat. Comput.\/}~{\em 10}, 197--208.

\bibitem[\protect\citeauthoryear{Doucet and Johansen}{Doucet and
  Johansen}{2009}]{doucet:johansen:2009}
Doucet, A. and A.~Johansen (2009).
\newblock A tutorial on particle filtering and smoothing: fifteen years later.
\newblock {\em Oxford handbook of nonlinear filtering\/}.

\bibitem[\protect\citeauthoryear{Fearnhead, Wyncoll, and Tawn}{Fearnhead
  et~al.}{2010}]{fearnhead:wyncoll:tawn:2010}
Fearnhead, P., D.~Wyncoll, and J.~Tawn (2010).
\newblock A sequential smoothing algorithm with linear computational cost.
\newblock {\em {B}iometrika\/}~{\em 97\/}(2), 447--464.

\bibitem[\protect\citeauthoryear{Gilks and Berzuini}{Gilks and
  Berzuini}{2001}]{gilks:berzuini:2001b}
Gilks, W.~R. and C.~Berzuini (2001).
\newblock Following a moving target---{M}onte {C}arlo inference for dynamic
  {B}ayesian models.
\newblock {\em J. Roy. Statist. Soc. B\/}~{\em 63\/}(1), 127--146.

\bibitem[\protect\citeauthoryear{Godsill, Doucet, and West}{Godsill
  et~al.}{2004}]{godsill:doucet:west:2004}
Godsill, S.~J., A.~Doucet, and M.~West (2004).
\newblock {M}onte {C}arlo smoothing for non-linear time series.
\newblock {\em J. Am. Statist. Assoc.\/}~{\em 99}, 156--168.

\bibitem[\protect\citeauthoryear{Hull and White}{Hull and
  White}{1987}]{hull:white:1987}
Hull, J. and A.~White (1987).
\newblock The pricing of options on assets with stochastic volatilities.
\newblock {\em J. Finance\/}~{\em 42}, 281--300.

\bibitem[\protect\citeauthoryear{Kitagawa}{Kitagawa}{1996}]{kitagawa:1996}
Kitagawa, G. (1996).
\newblock {M}onte-{C}arlo filter and smoother for non-{G}aussian nonlinear
  state space models.
\newblock {\em J. Comput. Graph. Statist.\/}~{\em 1}, 1--25.

\bibitem[\protect\citeauthoryear{Kitagawa}{Kitagawa}{1998}]{kitagawa:1998}
Kitagawa, G. (1998).
\newblock A self-organizing state-space model.
\newblock {\em J. Am. Statist. Assoc.\/}~{\em 93\/}(443), 1203--1215.

\bibitem[\protect\citeauthoryear{K\"{u}nsch}{K\"{u}nsch}{2000}]{kunsch:2000}
K\"{u}nsch, H.~R. (2000).
\newblock State space and hidden {M}arkov models.
\newblock In O.~E. Barndorff-Nielsen, D.~R. Cox, and C.~Kluppelberg (Eds.),
  {\em Complex Stochastic Systems}. CRC Press.

\bibitem[\protect\citeauthoryear{Liu and Chen}{Liu and
  Chen}{1998}]{liu:chen:1998}
Liu, J. and R.~Chen (1998).
\newblock Sequential {M}onte-{C}arlo methods for dynamic systems.
\newblock {\em J. Am. Statist. Assoc.\/}~{\em 93\/}(443), 1032--1044.

\bibitem[\protect\citeauthoryear{Olsson and Rydén}{Olsson and
  Rydén}{2010}]{olsson:ryden:2010}
Olsson, J. and T.~Rydén (2010).
\newblock Metropolising forward particle filtering backward sampling and
  rao-blackwellisation of metropolised particle smoothers.
\newblock Preprint, arXiv:1011.2153v1.

\end{thebibliography}

\end{document}